\documentclass[a4paper,10pt]{article}
\usepackage{float}

 \usepackage{latexstyle}
 \usepackage{bbm}
 \usepackage[para]{footmisc}

\usepackage{mathtools}
 
\usepackage{amsmath}
\makeatletter
\renewcommand*\env@matrix[1][*\c@MaxMatrixCols c]{%
  \hskip -\arraycolsep
  \let\@ifnextchar\new@ifnextchar
  \array{#1}}
\makeatother

\title{On Cartesian line sampling with anisotropic total variation regularization}

\author{Clarice Poon  \footnote{CEREMADE, Universit\'e Paris-Dauphine} \footnote{Email:\texttt{cmhsp2@cam.ac.uk}}}

\begin{document}
\maketitle

\begin{abstract}
This paper considers the use of the anisotropic total variation seminorm to recover a two dimensional vector $x\in \bbC^{N\times N}$ from its partial Fourier coefficients, sampled along Cartesian lines. We prove that if $(x_{k,j} - x_{k-1,j})_{k,j}$ has at most $s_1$ nonzero coefficients in each  column and $(x_{k,j} - x_{k,j-1})_{k,j}$ has at most $s_2$ nonzero coefficients in each row, then, up to multiplication by $\log$ factors, one can exactly recover $x$ by sampling along $s_1$ horizontal lines of its Fourier coefficients and along $s_2$ vertical lines of its Fourier coefficients. Finally, unlike standard compressed sensing estimates, the $\log$ factors involved are dependent on the separation distance between the nonzero entries in each row/column of the gradient of $x$ and not on $N^2$, the ambient dimension of $x$.
\end{abstract}

\section{Introduction}
Research into compressed sensing has resulted in several examples in which one can recover an $s$-sparse vector of length $N$ from $\ord{s \log N}$ randomly chosen linear measurements. One of the first  examples of this, by Cand\`{e}s, Romberg and Tao, is the recovery of a gradient sparse vector from samples of its Fourier transform by means of solving a total variation regularization problem. This is perhaps one of the most well known and influential results in compressed sensing because of its links with applications, in particular, this result motivated the use of total variation regularization to reduce the sampling cardinality in many imaging applications, such as Electron Microscopy \cite{leary2013compressed}, Magnetic Resonance Imaging (MRI) \cite{Lustig}, Optical Deflectometric Tomography \cite{gonzalez2014compressive}, Phase-Contrast Tomography \cite{cong2012differential} and  Radio Interferometry \cite{wiaux2009compressed}. However, while studies into uniformly random sampling provide some insight into how total variation regularization can allow one to subsample the Fourier transform, there are two further aspects that one should consider.

\paragraph{1. Dense sampling at low frequencies and sparsity structure.}

 It was observed in \cite{Lustig, lustig2007sparse} that one can obtain far superior results via \textit{variable density sampling} where one samples more densely at low frequencies. This effect is demonstrated in Figure \ref{fig:boat}, where we compare the reconstruction of the boat test image from 12.3\% of its Fourier coefficients via different sampling maps. On the theoretical side, one particular type of variable density sampling was first studied by Krahmer and Ward in \cite{krahmer2014stable} and later in \cite{tv_poon}. The analysis of \cite{tv_poon} showed that compared with sampling uniformly at random, one of the advantages offered by sampling more densely at low frequencies is improved robustness to inexact sparsity and noise. However, an important reason for the effectiveness of sampling densely at low frequencies is that although the sampling cardinality of $\ord{s \log N}$ is optimal for the recovery of $s$-sparse vectors, one can further reduce this sampling cardinality by placing a structure assumption on the vector to be recovered. This observation was made by Cand\`{e}s and Fernandez-Granda \cite{candes2014towards} in the context of recovering a superposition of diracs in super-resolution, and in the case of total variation regularization of one dimensional signals, by exploiting the results of \cite{candes2014towards} and \cite{tang2012compressive}, the following result was proved in \cite{tv_poon}:

\begin{theorem}\label{thm:1d}
Let $N\in\bbN$, let $\epsilon\in [0,1]$ and let $M\in\bbN$ be such that $N/4 \geq M\geq 10$. Let $A$ be the discrete Fourier transform on $\bbC^N$ (defined in Section \ref{sec:notation}).
\begin{itemize}
\item Let $x\in\bbC^N$ and  $\Delta \subset\br{1,\ldots, N}$ be of cardinality $s$ and suppose that $$\min_{k,j\in\Delta, k\neq j}\frac{\abs{k-j}}{N} \geq \frac{2}{M}.$$
\item  Let  $\Omega = \{0\}\cup\Omega'$ where $\Omega'\subset \br{-M,\ldots, M}$ consist of  $m$ indices chosen uniformly at random with
$$
m \gtrsim \max\br{ \log^2\left(\frac{M}{\epsilon}\right), \, s\cdot \log \left(\frac{s}{\epsilon}\right)\cdot \log\left(\frac{M}{\epsilon}\right)
}.$$
\end{itemize}

Then with probability exceeding $1-\epsilon$, given $y = P_\Omega A x + \eta$ and $\norm{\eta}_2\leq \delta\cdot\sqrt{m}$, any solution $\xi$ to 
\begin{equation}
\min_{x\in\bbC^N} \norm{x}_{TV} \text{ subject to } \norm{P_\Omega A x - y}_2\leq \delta\cdot \sqrt{m}
\end{equation} satisfies
$$
\frac{\norm{x-\xi}_2}{\sqrt{N}}
\lesssim \frac{N^2}{ M^2} \cdot \left(  \delta \cdot s+ \sqrt{s}\cdot \norm{P_{\Delta^c} D x}_1\right).
$$
If $m=2M+1$, then the error bound holds with probability 1.
\end{theorem}

So, if $x\in\bbC^N$  is $s$-gradient sparse with a minimum separation of $2/s$. Then, $x$ can be exactly recovered from $2s+1$ Fourier coefficients. However,  random sampling guarantees recovery only with $\ord{s\log N}$ samples. Thus, one can further reduce the number of samples required by choosing the samples in accordance to the underlying sparsity structure.

\paragraph{2. The need to understand more realistic sampling patterns.} Sampling in applications such as  MRI is constrained to sampling along smooth trajectories, such as radial lines, spirals or Cartesian lines \cite{lustig2007sparse,wang2009pseudo}; on the other hand, the majority of results in compressed sensing describe only the effects of pointwise sampling. To our knowledge, the only theoretical result in this direction is by Boyer et al.
 in \cite{boyer2015compressed} where they consider the use of wavelet regularization while sampling along horizontal (or vertical) Cartesian lines in the Fourier domain.

\subsection{This paper's contribution}
 The purpose of this paper is to present a two dimensional version of Theorem \ref{thm:1d}, where we  consider how one can efficiently sample the Fourier transform along Cartesian lines by taking into account the sparsity structure in the gradient of the underlying vector. The Cartesian sampling pattern studied in this paper is one of the sampling patterns which has been empirically studied in the application of compressed sensing to MRI \cite{lustig2007sparse,wang2009pseudo}. Thus, the result of this paper provides further justification and insight into the use of compressed sensing in MRI. The main result of this paper is presented and discussed in Section \ref{sec:main} and its proof is presented in Section \ref{sec:prf}.

\subsection{Related works -- wavelet regularization}
The link between success of dense sampling at low frequencies and the correspondence between such sampling patterns and the underlying sparsity structure has previously been investigated in the context of (orthogonal) wavelet regularization with Fourier sampling by Adcock et al. \cite{adcock2013breaking}. In the context of wavelet regularization, the relevant sparsity structure is the sparsity of the underlying wavelet coefficients within each wavelet scale and the results of \cite{adcock2013breaking} provide a link between the distribution of the Fourier samples and the  wavelet sparsity at each scale, while the result of \cite{gs_l1} demonstrate that one can recover the first $N$ wavelet coefficients of the lowest scales from $\ord{N}$ Fourier coefficients of the lowest frequencies by $\ell^1$ wavelet regularization. Thus, the notion that $\ell^1$ regularization can allow for sampling rates without the $\log$ factor is relevant in greater generality.

 In \cite{boyer2015compressed}, Boyer et al. investigated this structure dependency in the case of wavelet regularization with Cartesian line sampling in the Fourier domain. In particular, they proved that one can guarantee stable error bounds provided that the number of horizontal lines within each block of Fourier coefficients is proportional, up to log factors, with the sparsity in each column of the wavelet transform within the corresponding wavelet scale.

\begin{figure}

\begin{center}
\begin{tabular}{@{\hspace{0pt}}c@{\hspace{3pt}}c@{\hspace{3pt}}c}
\includegraphics[width = 0.2\textwidth]{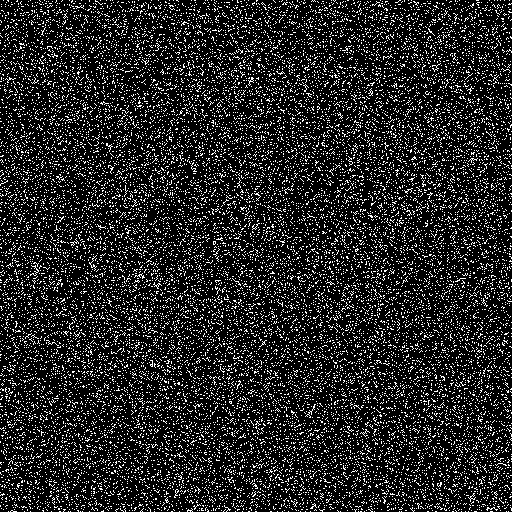}  &
\includegraphics[width = 0.2\textwidth]{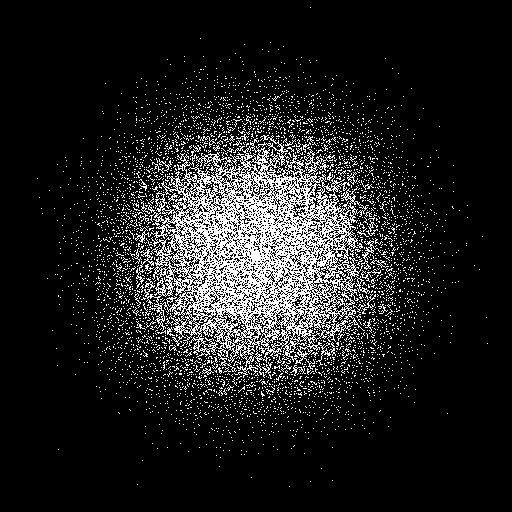}
&
\includegraphics[width = 0.2\textwidth]{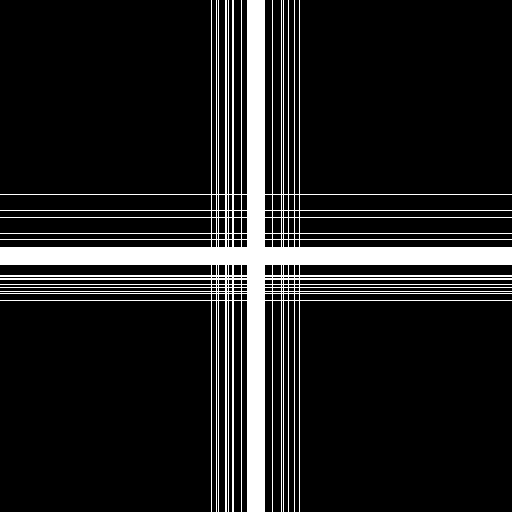} \\
  \includegraphics[width = 0.2\textwidth]{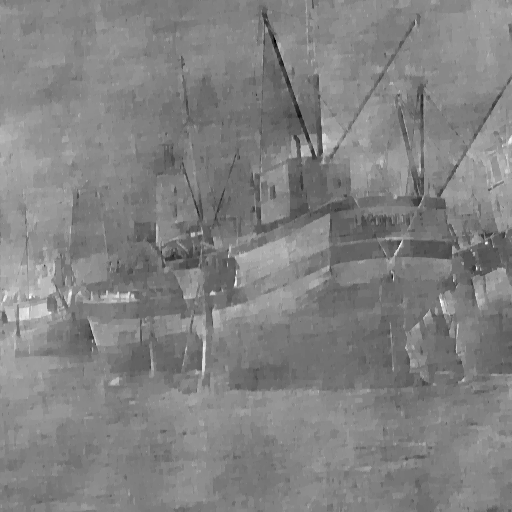} &\includegraphics[width = 0.2\textwidth]{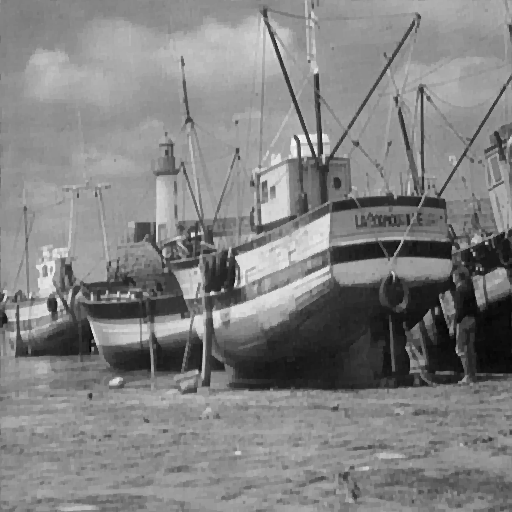}
&\includegraphics[width = 0.2\textwidth]{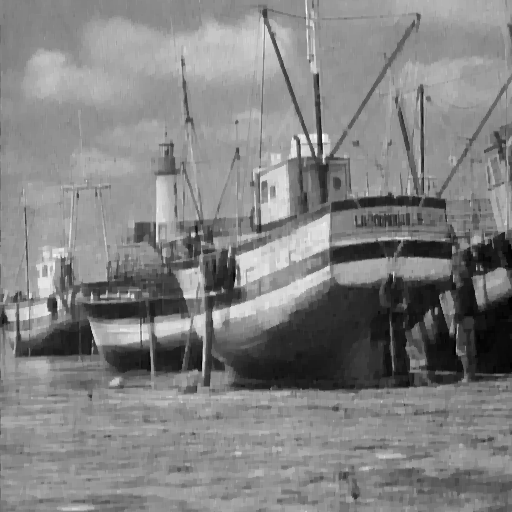}\\
Rel. Err. 28.9\% & Rel. Err. 6.6\% & Rel. Err. 7.6\%
\end{tabular}
\end{center}
\caption{The top row shows three different sampling maps, covering 12.3\% of the Fourier coefficients. Note that the zeroth Fourier frequency corresponds to the centre of each map. The bottom row shows the corresponding reconstructions. \label{fig:boat}}
\end{figure}

\subsection{Notation}\label{sec:notation}

Let $N\in\bbN$, and let $[N] = \br{-\lceil N/2 \rceil +1,\ldots, \lfloor N/2\rfloor}$. Let $A:\bbC^N\to \bbC^N$ be the discrete Fourier transform with
$$
A z = \left(\sum_{j=1}^N z_j e^{-2\pi i k j/N}\right)_{k\in [N]}, \qquad z=(z_j)_{j=1}^N \in\bbC^N.
$$
Let $D:\bbC^N\to \bbC^N$ be the finite differences operator defined for each $z\in\bbC^N$ as $D z = (z_j-z_{j-1})_{j=1}^N$ where $z_{0}:= z_N$.
Given any $\Lambda\subset \bbZ$ and $m\in\bbN$, $\Omega\sim \mathrm{Unif}(\Lambda, m)$ means that $\Omega$ consists of $m$ elements of $\Lambda$ drawn uniformly at random (without replacement).

Given any $z\in\bbC^N$, and any operator $V:\bbC^N\to \bbC^N$, $Vz$ is a column vector whenever $z$ is a column vector and a row vector whenever $z$ is a row vector.

Given any $z\in \bbC^{N\times N}$, let $z^{[\mathrm{col},j]} = (z_{k,j})_{k\in [N]}\in\bbC^N$ denote the $j^{\rth}$ column of $z$ and let  $z^{[\mathrm{row},k]} = (z_{k,j})_{j\in [N]}\in\bbC^N$ denote the $k^{\rth}$ row of $z$.

Let $\tilde A: \bbC^{N\times N}\to \bbC^{N\times N}$ with 
$$
\tilde A z = \left(\sum_{k_1=1}^N \sum_{k_2=1}^N z_{k_1,k_2} e^{-2\pi i (k_1 n_1+k_2n_2)/N}\right)_{n_1,n_2\in [N]}
$$
and let 
$
\tilde D_1 :  \bbC^{N\times N}\to \bbC^{N\times N}$ and $
\tilde D_2 :  \bbC^{N\times N}\to \bbC^{N\times N}$, with
$$
\tilde D_1 z = \left(z_{k,j}-z_{k-1,j}\right)_{k,j=1}^N, \quad \tilde D_2 z = \left(z_{k,j}-z_{k,j-1}\right)_{k,j=1}^N.
$$
Let $\tilde D z= (\tilde D_1 z, \tilde D_2 z)$,  $\norm{\tilde D z}_2 = \sqrt{\norm{\tilde D_1 z}_2^2 + \norm{\tilde D_2 z}_2^2}$, $\norm{\tilde D z}_1 = \norm{\tilde D_1 z}_1 + \norm{\tilde D_2 z}_1$.

Given any $ \Omega \subset \bbZ\times \bbZ$, let $\tilde P_{\Omega}: \bbC^{N\times N}\to \bbC^{N\times N}$ with
$$
\tilde P_\Omega z = y, \quad  y_j = \begin{cases}
z_j & j\in\Omega,\\
0 & j\not \in \Omega.
\end{cases}
$$
Let $\nm{\cdot}_{TV}$ denote the anisotropic total variation norm with
$$
\nm{z}_{TV} := \nm{\tilde D_1 z}_1+ \nm{\tilde D_2 z}_1, \qquad \forall z\in\bbC^{N\times N}
$$
and given $\Delta_1,\Delta_2\subset \br{1,\ldots, N}^2$, let
$$
\nm{z}_{TV, \Delta_1,\Delta_2} := \nm{\tilde P_{\Delta_1}\tilde D_1 z}_1+ \nm{\tilde P_{\Delta_2} \tilde D_2 z}_1, \qquad \forall z\in\bbC^{N\times N}.
$$

Given $a,b\in \bbR$, we write $a\lesssim b$ if there exists some constant $C>0$ (independent of all variables under consideration) such that $a\leq C\cdot b$.
\section{Key concepts}
In Theorem \ref{thm:1d}, the sparsity structure considered is the separation between the discontinuities of the underlying signal. When considering the recovery of some vector $x\in\bbC^{N\times N}$ by sampling along Cartesian lines of its Fourier transform, our main result will demonstrate how one should subsample depends on the sparsity and the minimum separation distance within each column of $\tilde D_1 x$ and each row of $\tilde D_2 x$. We first present three definitions that our main result will depend on.

\begin{definition}[Sparsity]
Let $\Delta \subset \br{1,\ldots, N}^2$. The column cardinality of $\Delta$ is  $$s_1 = \max_{j=1}^N \abs{\br{k: (k,j)\in\Delta}}.$$ The row  cardinality of $\Delta$ is $$s_2 = \max_{k=1}^N \abs{\br{j: (k,j)\in\Delta}}.$$
\end{definition}

\begin{definition}[Minimum separation distance]
Let $N\in\bbN$ and let $\Delta\subset \br{1,\ldots, N}^2$. The minimum separation distance of its rows is defined to be
$$
\nu_{\mathrm{row}}(\Delta, N) = \min_{n=1}^N \min\br{\frac{\abs{j-k}}{N} : (j,n), (k,n) \in \Delta, j\neq k},
$$
and minimum separation distance of its columns is defined to be
$$
\nu_{\mathrm{col}}(\Delta, N) = \min_{n=1}^N \min \br{\frac{\abs{j-k}}{N} : (n,j), (n,k)\in \Delta, j\neq k}.
$$

\end{definition}
\begin{definition}
We say that $x$ has $T_1$ distinct column supports if 
$$
 T_1 = \abs{\br{ (x_{k,j})_{k=1}^N :  j=1,\ldots, N}},
$$
and say that $x$ has $T_2$ distinct row supports if
$$
 T_2 = \abs{\br{ (x_{k,j})_{j=1}^N : k=1,\ldots, N}}.
$$
\end{definition}

\section{Main theorem}\label{sec:main}
Let $x\in\bbC^{N\times N}$ and let $\Delta_1,\Delta_2\subset \br{1,\ldots, N}^2$. Suppose that $\tilde P_{\Delta_1 } \sgn(\tilde D_1 x)$  has $T_1$ distinct column supports with a minimum separation of $2/M_1$ along its columns, and $\tilde P_{\Delta_2} \sgn(\tilde D_2 x)$ has $T_2$ distinct supports with a minimum separation of $2/M_2$ along  its rows. Suppose that $\Delta_1$ has column cardinality $s_1$ and $\Delta_2$ has row cardinality $s_2$.
Assume also that for $i=1,2$,
$$
s_i \log(T_i s_i/\epsilon) \geq \log(T_i M_i/\epsilon) .
$$
\begin{theorem}\label{thm:main}

 Let $\epsilon\in (0,1)$ and let
$
\Omega =\br{0}\cup \br{\Omega_1 \times [N]} \cup \br{[N]\times\Omega_2}$ and let $ m=\abs{\Omega},
$
where
 $$
\Omega_1 \sim \mathrm{Unif}( [M_1], m_1), \qquad m_1  \gtrsim s_1 \log(T_1 s_1/\epsilon)\log(T_1 M_1/\epsilon), $$
 and
 $$ \Omega_2 \sim \mathrm{Unif}( [M_2], m_2), \qquad m_2  \gtrsim  s_2 \log(T_2 s_2/\epsilon)\log(T_2 M_2/\epsilon).
 $$ 
 Let $\xi = P_{\Omega} x + \eta$ with $\nm{\eta}_2\leq \delta\cdot \sqrt{m}$, and suppose that $\hat x$ is a minimizer of
\be{\label{eq:tv_min_2d}
\min_{z\in\bbC^{N\times N}} \nm{z}_{TV} \text{ subject to } \nm{\tilde P_\Omega \tilde A z - \xi}_2\leq \delta \cdot \sqrt{m}.
}
Then, with probability exceeding $1-\epsilon$,
\be{\label{thm:err1}
\nm{\tilde D (x-\hat x)}_2  \lesssim \frac{N^2}{M_0^2}\left((m_0 N)^{-1/2} \sqrt{m} \delta  
+  \nm{x}_{TV, \Delta_1^c, {\Delta_2}^c  }\right),
}
 and
\be{\label{thm:err2}
\nm{x-\hat x}_2 \lesssim \frac{N^2}{M_0^2}\left((m/m_0)^{1/2} \sqrt{s } \delta +  \sqrt{s }   \nm{x}_{TV, {\Delta_1}^c,  {\Delta_2}^c } \right),
}
where
$s=\max\br{s_1,s_2}$, $m_0=\min\br{m_1,m_2}$, and $M_0=\min\br{M_1,M_2}$. 
If $\Omega_1=[M_1]$ and $\Omega_2=[M_2]$, then (\ref{thm:err1}) and (\ref{thm:err2}) hold with probability one.

\end{theorem}

\subsection{Remarks on the main result}

\paragraph{Removal of the $\log N$ factor:} Suppose that $x\in\bbC^{N\times N}$ is such that the support of $\tilde D^{[\mathrm{col}]} x$ consists of $M_1$ lines which have a minimum separation of $1/M_1$ and the support of $\tilde D^{[\mathrm{row}]} x$ consists of $M_2$ lines which have a minimum separation of $1/M_2$. Then, one can perfectly recover $x$ from sampling the Fourier transform of $x$ along $\ord{M_1}$ horizontal lines and $\ord{M_2}$ vertical lines. On the other hand, in this case, the sparsity is $(M_1+M_2)N$ and one is guaranteed exact recovery from sampling uniformly at random when one observes $\ord{(M_1+M_2)N\log N}$ samples. Figure \ref{fig:unif_vs_struct} illustrates this effect by showing the recovery of an image from $1.2\%$ of its Fourier coefficients. The test image can be perfect reconstructed from sampling the low frequency Cartesian lines, on the other hand, sampling 1.2 \% of the Fourier coefficients uniformly at random yields a poor reconstruction.  Note that this image can in fact be recovered from 3\% of its Fourier coefficients drawn uniformly at random, it is simply that one can sample less by considering the sparsity structure of the test image.

\paragraph{The importance of structure dependency:} Note that from Theorem \ref{thm:main}, the range that one should sample from is dependent on the minimum separation between the discontinuities in the corresponding direction, and the number of samples that one should draw is up to $\log$ factors dependent on the maximum sparsity in each row or each column of the corresponding direction. It is important to consider this structure dependency when devising a sampling scheme -- see Figure \ref{fig:struct_matters}.

\begin{figure}

\begin{center}
\begin{tabular}{@{\hspace{0pt}}c@{\hspace{1pt}}c@{\hspace{3pt}}c@{\hspace{1pt}}c}
\includegraphics[width = 0.2\textwidth]{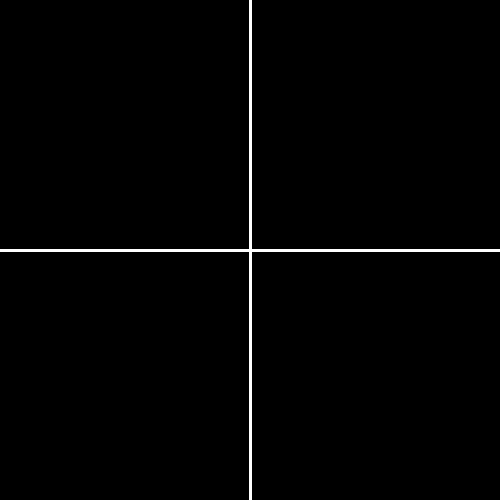}  &
  \includegraphics[width = 0.2\textwidth]{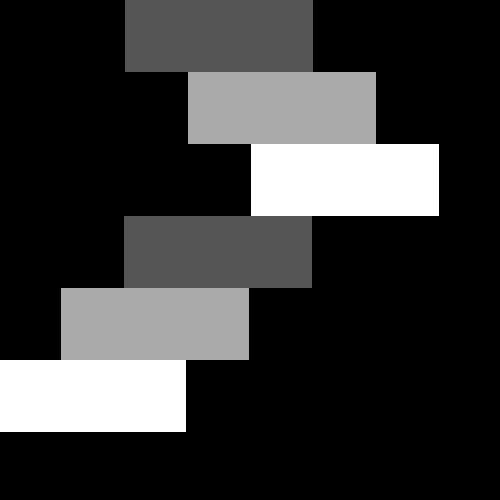} &
\includegraphics[width = 0.2\textwidth]{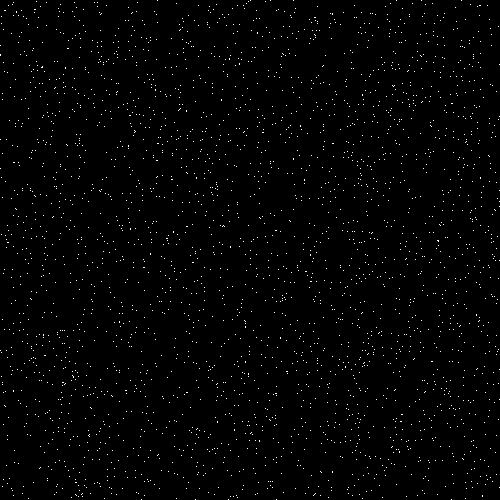}
&
\includegraphics[width = 0.2\textwidth]{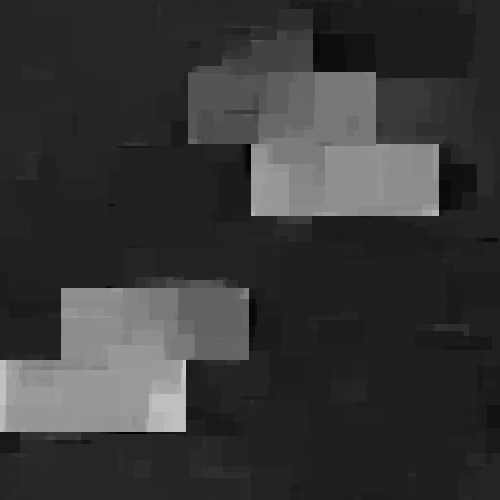} \\
(a) Samp. map &  (a) Recov. & (b)  Samp. map & (b) Recov.\\
& Rel Err. 0\% && Rel. Err. 49.78\%
\end{tabular}
\end{center}
\caption{In (a), the test image (500$\times 500$) can be perfectly recovered from 1.2\% of its Fourier coefficients, indexed by the Cartesian lines passing through the low frequencies. In (b), the reconstruction obtained from sampling 1.2\% of the Fourier coefficients uniformly at random is shown. \label{fig:unif_vs_struct}}
\end{figure}

\begin{figure}

\begin{center}
\begin{tabular}{@{\hspace{0pt}}c@{\hspace{1pt}}c@{\hspace{3pt}}c@{\hspace{1pt}}c}
\includegraphics[width = 0.2\textwidth]{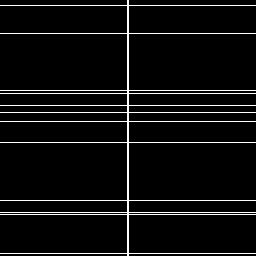}  &
 \includegraphics[width = 0.2\textwidth]{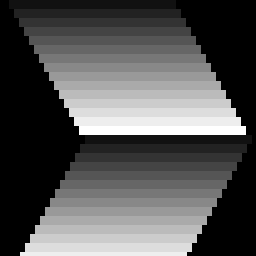} 
&
\includegraphics[width = 0.2\textwidth]{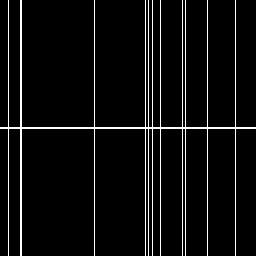}&
\includegraphics[width = 0.2\textwidth]{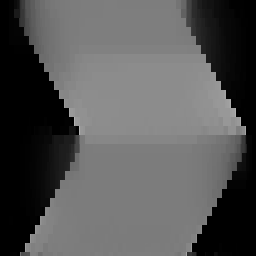} \\
 (a) Samp. map &  (a) Recov. & (b) Samp. map &  (b) Recov.\\
&Rel. Err. 0\% & & Rel. Err. 44.8\%
\end{tabular}
\end{center}
\caption{ This figure shows two sampling maps and their corresponding reconstructions (256$\times$256). Both sampling maps cover 5.4\% of the Fourier coefficients. The sampling map (a) is constructed by sampling along 12 lines uniformly at random in the horizontal direction and along the 2 lowest frequency lines in the vertical direction. The sampling map (b) is obtained in the same manner, but sampled in the opposite orientations with 12 random lines in the vertical direction and 2 of the lowest frequency lines in the horizontal direction.
\label{fig:struct_matters}}
\end{figure}

\section{Proof}\label{sec:prf}
As is now standard in compressed sensing, the  proof of Theorem \ref{thm:main} consists of showing the existence of some dual certificate \cite{candes2006robust,candes2011probabilistic}.  Following the arguments in \cite{tv_poon}, one can show that given $x\in\bbC^{N\times N}$ with $\mathrm{supp}(\tilde D_1 x) = \Delta_1$ and $\mathrm{supp}(\tilde D_2 x) = \Delta_2$,  $x$ is the unique solution of (\ref{eq:tv_min_2d}) if $0\in\Omega$, and the following two conditions hold:
\begin{itemize}
\item[(i)] For $i=1,2$, there exists some dual certificate $\rho_i,\in\mathrm{ran}(\tilde A^* \tilde P_\Omega)$ such that $\norm{\rho_i}_{\infty}$ and $\tilde P_{\Delta_i} \rho_i = \mathrm{sgn}(\tilde D_i x)$,
\item[(ii)] $\tilde P_\Omega \tilde A \tilde P_{\Delta_i}$ is injective for $i=1,2$.
\end{itemize}
However, instead of directly showing the existence of one dual certificate defined on $\bbC^{N\times N}$, we will exploit the fact that we are given Fourier samples along Cartesian lines and apply Proposition \ref{prop:dual} to  show that it suffices to prove the existence a sequence of one dimensional certificates defined on $\bbC^N$.

\begin{lemma}\label{lem:fourier_id}
Let $N\in\bbN$, $\Omega \subset [N]$ and let $z\in\bbC^{N\times N}$.
Then,
$$
\nm{\tilde P_{\Omega \times [N]} \tilde A z}_2^2 = N \sum_{k=1}^N \nm{P_{\Omega }A z^{[\mathrm{col},k]}}_2^2,
\qquad
\nm{\tilde P_{[N]\times \Omega } \tilde A z}_2^2 = N \sum_{k=1}^N \nm{P_{\Omega }A z^{[\mathrm{row},k]}}_2^2.
$$
\end{lemma}
\prf{
We first note that
\eas{
(\tilde A z)_{n_1,n_2} &= \sum_{k_2=1}^N \left(\sum_{k_1=1}^N z_{k_1,k_2} e^{-2\pi i  k_1 n_1 /N}\right) e^{-2\pi i  k_2 n_2 /N} \\
&=  \sum_{k_2=1}^N (A z^{[\mathrm{col}, k_2]})_{n_1} e^{-2\pi i  k_2 n_2 /N} = (A \beta^{[n_1]} )_{n_2},
}
where $\beta^{[n_1]} = \left(( A z^{[\mathrm{col}, k]})_{n_1}\right)_{k=1}^N$.
 Then, by using this identity, we have that
\spl{\label{eq:p1_bd1}
\nm{\tilde P_{\Omega \times [N]} \tilde A z}_2^2 &= \sum_{n_1\in\Omega }\sum_{n_2\in [N]} \abs{(\tilde A z)_{n_1,n_2}}^2
= \sum_{n_1\in\Omega }\sum_{n_2\in [N]} \abs{(A \beta^{[n_1]})_{n_2}}^2\\
&=\sum_{n_1\in\Omega } \nm{A \beta^{[n_1]}}_2^2 =  N\sum_{n_1\in\Omega}\nm{\beta^{[n_1]}}_2^2
=  N \sum_{k=1}^N \sum_{n_1\in\Omega } \abs{( A z^{[\mathrm{col}, k]})_{n_1}}^2\\
&= N \sum_{k=1}^N \nm{P_{\Omega }A z^{[\mathrm{col},k]}}^2_2,
}
where we have applied in the second line the fact that $N^{-1/2}A$ is unitary. Finally, by a symmetric argument,
$$
\nm{\tilde P_{[N]\times \Omega } \tilde A z}_2^2 = N \sum_{k=1}^N \nm{P_{\Omega }A z^{[\mathrm{row},k]}}_2^2.
$$
}

\begin{proposition}[Dual certificates] \label{prop:dual}
Let $x\in\bbC^{N\times N}$.
Let $\Delta_1,\Delta_2\subset \br{1,\ldots, N}^2$.
 Let $\Omega_1,\Omega_2\subset [N]$ and let $m_1=\abs{\Omega_1}$ and $m_2=\abs{\Omega_2}$. Let
$$
\Omega = \br{\Omega_1 \times [N]} \cup \br{[N]\times\Omega_2}.
$$
Let $m=\abs{\Omega}$. Let $\xi = P_{\Omega} x + \eta$ with $\nm{\eta}_2\leq \delta\cdot \sqrt{m}$, and suppose that $\hat x$ is a minimizer of
\be{\label{eq:prop_min}
\min_{z\in\bbC^{N\times N}} \nm{z}_{TV} \text{ subject to } \nm{\tilde P_\Omega \tilde A z - \xi}_2\leq \delta \cdot \sqrt{m}.
}
For $n=1,2$, let
 $\Delta_{n,j} =  \br{k: (k,j)\in\Delta_n}\subset \br{1,\ldots, N}$ and let $s_n = \max_{j=1}^N \abs{\Delta_{n,j}}$. 
Assume that the following conditions hold.
\begin{itemize}
\item[(i)] For each $j\in \br{1,\ldots, N}$, $$m_1^{-1/2} \inf_{\mathrm{supp}(x) = \Delta_{1,j}, \nm{x}_2 = 1}\nm{P_{\Omega_1} A x}_2  \geq c_1>0.
$$

\item[(ii)]   For each $j\in \br{1,\ldots, N}$, $$m_2^{-1/2} \inf_{\mathrm{supp}(x) = \Delta_{2,j}, \nm{x}_2 = 1} \nm{P_{\Omega_2} A x}_2  \geq c_1>0.
$$

\item[(iii)]  For each $j\in \br{1,\ldots, N}$, there exists $\rho_j = m_1^{-1/2} A^*P_{\Omega_1} w_j\in\bbC^N$ such that
$$
P_{\Delta_{1,j}} \rho_j = P_{\Delta_{1,j}} \sgn(\tilde D_1 x)^{[\mathrm{col},j]}, \qquad \nm{P_{\Delta_{1,j}}^\perp \rho_j}_\infty \leq c_2 <1, \qquad \sum_{j=1}^N\nm{w_j}_2\leq L^2.
$$
\item[(iv)]  For each $j\in\br{1,\ldots, N}$, there exists $\tau_j = m_2^{-1/2} A^*P_{\Omega_2} u_j\in\bbC^N$ such that
$$
P_{\Delta_{2,j}} \tau_j = P_{\Delta_{2,j}} \sgn(\tilde D_2 x)^{[\mathrm{col},j]}, \qquad \nm{P_{\Delta_{2,j}}^\perp \tau_j}_\infty \leq c_2 <1, \qquad \sum_{j=1}^N\nm{u_j}_2\leq L^2.
$$
\end{itemize}
Then,
$$
\nm{\tilde D (x-\hat x)}_2  \lesssim C_1 \cdot \delta \cdot 
+ C_2 \cdot \nm{x}_{TV, {\Delta_1}^c, {\Delta_2}^c  },
$$
 and
$$
\nm{x-\hat x}_2 \lesssim C_3\cdot \delta\cdot \sqrt{\frac{sm}{m_0} }  + C_2\cdot\sqrt{s } \cdot  \nm{x}_{TV, {\Delta_1}^c,  {\Delta_2}^c },
$$
where  
$s=\max\br{s_1,s_2}$, $m_0=\min\br{m_1,m_2}$,  
$$
C_1 =  (m_0 N)^{-1/2} m^{1/2}    (1+ L) (1+ c_1^{-1} )(1-c_2)^{-1}    , \qquad  C_2 =  (1+ c_1^{-1} ) (1-c_2)^{-1},
$$
and
$C_3 =
 c_1^{-1}  
( 1    +  L
(1-c_2)^{-1}      N ^{-1/2}  ).
$

\end{proposition}

\begin{proof}[Proof of Proposition \ref{prop:dual}]

First, suppose that $z\in\bbC^{N\times N}$ is such that $ \nm{\tilde P_{\Omega_1\times [N]}\tilde A z}_2 \leq \sqrt{m}\delta$.
By applying assumption (i) and Lemma \ref{lem:fourier_id},
\eas{&
c_1^2\nm{\tilde P_{\Delta_1} z}_2^2 = c_1^2 \sum_{j=1}^N \nm{P_{\Delta_{1,j}} z^{[\mathrm{col},j]}}_2^2  \leq \sum_{j=1}^N m_1^{-1} \nm{P_{\Omega_1} A P_{\Delta_{1,j}} z^{[\mathrm{col},j]}}_2^2 \\
&= (m_1 N)^{-1} \nm{\tilde P_{\Omega_1\times [N]}\tilde A \tilde P_{\Delta_1} z}_2^2 
 \leq (m_1 N)^{-1} \left( \nm{\tilde P_{\Omega_1\times [N]}\tilde A z}_2 + \nm{\tilde P_{\Omega_1\times [N]}\tilde A \tilde P_{\Delta_1}^\perp z  }_2 \right)^2\\
&\leq  (m_1 N)^{-1} \left( \sqrt{m}\delta + \max_{j\not \in \Delta_1}\nm{\tilde P_{\Omega_1\times [N]}\tilde A e_j}_2 \nm{\tilde P_{\Delta_1}^\perp z  }_1 \right)^2.
}
Note that
\eas{
&(m_1 N)^{-1/2}\max_{j\not \in \Delta_1}\nm{\tilde P_{\Omega_1\times [N]}\tilde A e_j}_2 =1.
}
It thus follows that 
\be{\label{p1eq:1}
\nm{\tilde P_{\Delta_1} z}_2 \leq \frac{\sqrt{m}\delta}{c_1 \sqrt{N m_1}} + \frac{  \nm{\tilde P_{\Delta_1}^\perp z  }_1}{c_1}.
}
Similarly, it follows from assumption (ii) that for any  $z\in\bbC^{N\times N}$ with $ \nm{\tilde P_{ [N]\times \Omega_2}\tilde A z}_2 \leq \sqrt{m}\delta$,
\be{\label{p1eq:2}
 \nm{\tilde P_{\Delta_2} z}_2 \leq  \frac{\sqrt{m}\delta}{c_1\sqrt{N m_2}} + \frac{ \nm{\tilde P_{\Delta_2}^\perp z  }_1}{c_1}.
}

Let $h=\hat x -x$, and observe that since $\hat x$ and $x$ both satisfy the constraint of \eqref{eq:prop_min}, $$\nm{\tilde P_{\Omega_1\times [N]} \tilde A h}_2 \leq \nm{\tilde P_{\Omega_1\times [N]} \tilde A x - \xi}_2+ \nm{\tilde P_{\Omega_1\times [N]}\tilde A  \hat x - \xi}_2\leq 2\sqrt{m}\delta.$$ Note also that
$$
(\tilde A \tilde D_1 h)_{k,j} = (1-e^{-2\pi i k/N}) (\tilde A h)_{k,j}, \qquad (\tilde A \tilde D_2 h)_{k,j} = (1-e^{-2\pi i j/N}) (\tilde A h)_{k,j}.
$$
Therefore, $\nm{\tilde P_{\Omega_1\times [N]} \tilde A \tilde D_1 h}_2 \leq 2\nm{\tilde P_{\Omega \times [N]} \tilde A h}_2 \leq 4\delta\sqrt{m}$.
Similarly,
$\nm{\tilde P_{[N]\times \Omega_2}\tilde A \tilde D_2 h}_2 \leq 4\delta\sqrt{m}$.
 So, we can apply the bounds (\ref{p1eq:1}) and (\ref{p1eq:2}) to obtain
\spl{\label{p1eq:3}
\nm{\tilde P_{\Delta_1} \tilde D_1 h}_2 \leq \frac{4\sqrt{m}\delta}{c_1 \sqrt{N m_1}} + \frac{ 4 \nm{\tilde P_{\Delta_1}^\perp \tilde D_1 h  }_1}{c_1},\\
 \nm{\tilde P_{\Delta_2 } \tilde D_2 h}_2 \leq  \frac{4\sqrt{m}\delta}{c_1\sqrt{N m_2}} + \frac{4 \nm{\tilde P_{\Delta_2 }^\perp \tilde D_2 h  }_1}{c_1}.
}

We now proceed to derive upper bounds for $\nm{\tilde P_{\Delta_1}^\perp \tilde D_1 h  }_1$ and $ \nm{\tilde P_{\Delta_2}^\perp \tilde D_2 h  }_1$: By H\"{o}lder's inequality and the triangle inequality,
\eas{
&\nm{\tilde D_1 \hat x}_1 = \nm{P_{\Delta_1}\tilde D_1 (x+h)}_1 + \nm{P_{\Delta_1}^\perp \tilde D_1 (x+h)}_1\\
&\geq \nm{P_{\Delta_1}\tilde D_1 x}_1 + \Re \ip{P_{\Delta_1}\tilde D_1 h}{\sgn(\tilde D_1 x)}
+ \nm{P_{\Delta_1}^\perp \tilde D_1 h}_1 - \nm{P_{\Delta_1}^\perp \tilde D_1 x}_1.
}
Rearranging the terms yields
$$
\nm{P_{\Delta_1}^\perp \tilde D_1 h}_1  \leq \nm{\tilde D_1 \hat x}_1 - \nm{\tilde D_1 x}_1 +2\nm{P_{\Delta_1}^\perp \tilde D_1 x}_1 + \abs{\ip{P_{\Delta_1}\tilde D_1 h}{\sgn(\tilde D_1 x)}}.
$$
Similarly,
$$
\nm{P_{\Delta_2}^\perp \tilde D_2 h}_1  \leq \nm{\tilde D_2\hat x}_1 - \nm{\tilde D_2 x}_1 +2\nm{P_{\Delta_2}^\perp \tilde D_2 x}_1 + \abs{\ip{P_{\Delta_2}\tilde D_2  h}{\sgn(\tilde D_2 x)}}.
$$
Now, since $\hat x$ is a minimizer of \eqref{eq:prop_min}, 
$$
\nm{\tilde D_1 \hat x}_1+ \nm{\tilde D_2 \hat x}_1 \leq  \nm{\tilde D_1 x}_1 + \nm{\tilde D_2 x}_1.
$$
So,
\spl{\label{p1eq:4}
&\nm{P_{\Delta_1}^\perp \tilde D_1 h}_1 + \nm{P_{\Delta_2 }^\perp \tilde D_2 h}_1   \\&\leq
2 \nm{x}_{TV, \Delta_1^c, \Delta_2^c  } + \abs{\ip{P_{\Delta_1}\tilde D_1 h}{\sgn(\tilde D_1 x)}} +  \abs{\ip{P_{\Delta_2 }\tilde D_2  h}{\sgn(\tilde D_2 x)}}.
}
We now proceed to bound $\abs{\ip{P_{\Delta_1}\tilde D_1  h}{\sgn(\tilde D_1 x)}}$. 
Let $y = \tilde D_1 x$ and let $z = \tilde D_1 h$. By using the existence of $\rho_j =  m_1^{-1/2} A^*P_{\Omega_1} w_j\in\bbC^N$ for $j=1,\ldots, N$ (from assumption (iii)), we have the following bound.
\eas{
&\abs{\ip{\tilde P_{\Delta_1}z}{\sgn(y)}} 
= \abs{\sum_{j=1}^N \ip{P_{\Delta_{1,j}} z^{[\mathrm{col},j]}}{ \sgn(y)^{[\mathrm{col},j]}}  }\\
&= \abs{\sum_{j=1}^N \ip{P_{\Delta_{1,j}} z^{[\mathrm{col},j]}}{ \sgn(y)^{[\mathrm{col},j]} - \rho_j} + \ip{z^{[\mathrm{col},j]}}{\rho_j} - \ip{P_{\Delta_{1,j}}^\perp z^{[\mathrm{col},j]}}{\rho_j}  }\\
&\leq   \abs{ \sum_{j=1}^N \ip{ m_1^{-1/2} P_{\Omega_1} A z^{[\mathrm{col},j]}}{w_j}} +  \sum_{j=1}^N \nm{P_{\Delta_{1,j}}^\perp z^{[\mathrm{col},j]}}_1 \nm{P_{\Delta_{1,j}}^\perp \rho_j}_\infty  \\
&\leq \sqrt{\sum_{j=1}^N m_1^{-1} \nm{P_{\Omega_1} A z^{[\mathrm{col},j]}}_2^2}\sqrt{ \sum_{j=1}^N \nm{w_j}_2^2} + c_2 \nm{\tilde P_{\Delta_1}^\perp z}_1,
}
where we have applied the Cauchy-Schwarz inequality to obtain the last line.
Recall from Lemma \ref{lem:fourier_id} that 
\eas{
&\sqrt{\sum_{j=1}^N m_1^{-1}\nm{P_{\Omega_1} A z^{[\mathrm{col},j]}}_2^2} = (m_1 N)^{-1/2} \nm{\tilde P_{\Omega_1\times [N]}\tilde A z}_2 \\
&=  (m_1 N)^{-1/2} \nm{\tilde P_{\Omega_1\times [N]}\tilde A D h}_2\leq 4 (m_1 N)^{-1/2} m^{1/2} \delta.
}
Hence it follows that
$$
\abs{\ip{\tilde P_{\Delta_1} \tilde D_1 h}{\sgn(\tilde D_1 x)}} \leq 4 L (m_1 N)^{-1/2} m^{1/2} \delta + c_2 \nm{\tilde P_{\Delta_1}^\perp \tilde D_1 h}_1.
$$
A similar argument also yields
$$
\abs{\ip{\tilde P_{\Delta_2}\tilde D_2 h}{\sgn(\tilde D_2 x)}} \leq 4 L (m_2 N)^{-1/2} m^{1/2} \delta + c_2 \nm{\tilde P_{\Delta_2}^\perp \tilde D_2 h}_1.
$$
By plugging these two estimates back into (\ref{p1eq:4}) and rearranging, we have that
\spl{\label{p1eq:5}
&\nm{\tilde P_{\Delta_1}^\perp \tilde D_1 h}_1 + \nm{\tilde P_{\Delta_2}^\perp \tilde D_2 h}_1
\\
&\leq (1-c_2)^{-1} \left( 4 L N^{-1/2} m^{1/2} (m_1^{-1/2}+m_2^{-1/2}) \delta + \nm{x}_{TV, \Delta_1^c,  \Delta_2^c } \right).
}
Combining this estimate with (\ref{p1eq:3}) yields the required bound on $\nm{\tilde D (\hat x - x)}_2$. 

To derive the bound on $\nm{\hat x-x}_{TV}$, first let $z= \tilde D_1 (\hat x - x)$. Note that by the Cauchy-Schwarz inequality,
\eas{
&\nm{\tilde P_{\Delta_1} z}_1 =
\sum_{j=1}^N \nm{ P_{\Delta_{1,j}} z^{[\mathrm{col},j]}}_1 \leq \sqrt{s_1} \sum_{j=1}^N \nm{ P_{\Delta_{1,j}} z^{[\mathrm{col},j]}}_2}
By applying condition (i),  $\nm{\tilde P_{\Delta_1} z}_1 $ is upper bounded by
\eas{
&
  \frac{\sqrt{s_1}}{c_1 \sqrt{m_1}} \sum_{j=1}^N \nm{P_{\Omega_1} A  P_{\Delta_{1,j}} z^{[\mathrm{col},j]}}_2 \leq  \frac{\sqrt{s_1}}{c_1 \sqrt{m_1}} \sum_{j=1}^N \left(\nm{P_{\Omega_1} A   z^{[\mathrm{col},j]}}_2 + \nm{P_{\Omega_1} A   P_{\Delta_{1,j}}^\perp z^{[\mathrm{col},j]}}_2 \right)\\
&\leq \frac{\sqrt{s_1}}{c_1 \sqrt{m_1}} \sqrt{N}\sqrt{\sum_{j=1}^N \nm{P_{\Omega_1} A   z^{[\mathrm{col},j]}}_2^2} + \frac{\sqrt{s_1}}{c_1 \sqrt{m_1}}\max_{j=1}^N\max_{l\not\in \Delta_{1,j}}\nm{P_{\Omega_1} A  e_l}_2 \sum_{j=1}^N \nm{ P_{\Delta_{1,j}}^\perp z^{[\mathrm{col},j]}}_1\\
&\leq \frac{\sqrt{s_1}}{c_1 \sqrt{m_1}} \nm{\tilde P_{\Omega_1\times [N]}\tilde A z}_2   + \frac{\sqrt{s_1}}{c_1  }  \nm{ \tilde P_{\Delta_1}^\perp z }_1
\leq  \frac{4\delta \sqrt{m}\sqrt{s_1}}{c_1 \sqrt{m_1}} + \frac{\sqrt{s_1}}{c_1}  \nm{ \tilde P_{\Delta_1 }^\perp z }_1.
}
By a symmetric argument,
$$
\nm{\tilde P_{\Delta_2}z}_1 \leq \frac{4\delta \sqrt{m}\sqrt{s_2}}{c_1 \sqrt{m_2}} + \frac{\sqrt{s_2}}{c_1}  \nm{ \tilde P_{\Delta_2}^\perp z }_1
$$
Therefore, by combining with the bound from (\ref{p1eq:5}),
$$
\nm{\hat x - x}_{TV} \lesssim C \frac{\sqrt{s m} }{\sqrt{m_0}}\delta + \left(1+\frac{\sqrt{s }}{c_1}\right)(1-c_2)^{-1} \nm{x}_{TV, {\Delta_1}^c,  {\Delta_2}^c }  ,
$$
where $s=\max\br{s_1,s_2}$, $m_0=\min\br{m_1,m_2}$, and
$C= c_1^{-1}  
( 1    +  L
(1-c_2)^{-1}      N ^{-1/2}  ).
$
Finally, recall that due to the Poincar\'{e} inequality, any zero mean image $X\in \bbC^{N\times N}$ satisfies
\be{\label{poinc}
\norm{X}_2 \leq \norm{X}_{TV}.
}
Since $0\in\Omega$, $\abs{\sum_j (x-\hat x)_j}\leq 2\delta\sqrt{m}$, so, by letting $X_j = (\hat x - x)_j - \frac{1}{N^2}\sum_j (x-\hat x)_j$ for each $j\in\br{1,\ldots,N}^2$,  (\ref{poinc}) and  the triangle inequality yields 
$$
\norm{\hat x - x}_2 \leq \delta + \nm{\hat x - x}_{TV},
$$
and hence, the conclusion follows.

\end{proof}

\begin{proof}[Proof of Theorem \ref{thm:main}]
To prove this theorem, we simply need to show that conditions (i) to (iv) of Proposition \ref{prop:dual} hold with high probability.
Note that these conditions  were studied in \cite{tv_poon}: we recall from Lemmas 4.25 of \cite{tv_poon}  that given any $\Delta\subset \br{1,\ldots, N}$  with a minimum separation distance of $1/M$, if $\Omega\subset [M]$ consists $m$ indices chosen uniformly at random with 
$$
m \gtrsim \max\br{\log^2(M/\epsilon), \abs{\Delta}\log(\abs{\Delta}/\epsilon)\log(M/\epsilon)},
$$
 then the following hold with probability exceeding $1-\epsilon$.
\begin{enumerate}
\item For all $x\in\bbC^N$, $$m^{-1/2}\nm{P_{\Omega } A P_{\Delta } x}_2  \geq \frac{3}{2\sqrt{5}}  \nm{P_{\Delta } x}_2 ,
$$
\item There exists $\rho = m^{-1/2} A^* P_\Omega w$ with $\nm{w} \lesssim \sqrt{\abs{\Delta}}$ such that
$$
P_\Delta \rho = x_0, \qquad \nm{P_\Delta^\perp \rho}_\infty \leq c(M) 
$$
where
$$
c(M) := \max\br{0.99993, 1-\frac{0.92(M^2-1)}{N^2}}.
$$
\end{enumerate}
Furthermore, these two conditions hold with probability 1 if $\Omega = [M_1]$.

Therefore, by applying the above fact $T_1$ times and applying the union bound, conditions (i) and (iii) of Proposition \ref{prop:dual} (with $c_1^{-1} = 2\sqrt{5}/3$,  $c_2 = c(M_1)$ and $L\lesssim\sqrt{s_1}$ ) hold with probability exceeding $1-T_1\epsilon$ provided that $\Omega_1$ is chosen uniformly at random with
\be{\label{Om1_c}
\Omega_1\subset [M_1], \qquad \abs{\Omega_1} \gtrsim \max\br{\log^2(M_1/\epsilon),  s_1 \log(s_1/\epsilon)\log(M_1/\epsilon)},
}
and they hold with probability 1 if $\Omega_1 = [M_1]$. Similarly, conditions (ii) and (iv) of Proposition \ref{prop:dual} (with $c_1^{-1} = 2\sqrt{5}/3$,  $c_2 = c(M_2)$ and $L\lesssim\sqrt{s_2}$ ) hold  with probability exceeding $1-T_2\epsilon$ if $\Omega_2$ is chosen uniformly at random with
\be{\label{Om2_c}
\Omega_2\subset [M_2], \qquad \abs{\Omega_2} \gtrsim \max\br{\log^2(M_2/\epsilon),   s_1 \log(s_2/\epsilon)\log(M_2/\epsilon)},
}
and they hold with probability 1 if $\Omega_2 = [M_2]$.
So, by applying the union bound once more, conditions (i) to (iv) with   $c_1^{-1} = 2\sqrt{5}/3$,  $c_2 = \max\br{c(M_1), c(M_2)}$ and $L\lesssim \max\br{\sqrt{s_1},\sqrt{s_2}} $ are satisfied with probability exceeding $1-T_1\epsilon-T_2\epsilon$ provided that $\Omega_1$ and $\Omega_2$ are chosen uniformly at random such that (\ref{Om1_c}) and (\ref{Om2_c}) hold.
\end{proof}

\section{Conclusion}
In this paper, we have derived recovery guarantees for total variation regularized solutions when given partial measurements of the Fourier transform taken along Cartesian lines. In particular, we established a link between the sparsity structure and the sampling pattern by proving that the number of Cartesian lines required for accurate recovery is dependent not only on the gradient sparsity  of the underlying vector, but also on the separation distance between the discontinuities.

\section{Acknowledgements}
The author acknowledges support from the Fondation Science Mathematique de Paris. 

 \addcontentsline{toc}{section}{References}
\bibliographystyle{abbrv}
\bibliography{References}

\end{document}